\documentclass[12pt]{article}
\usepackage{braket,amsfonts,amsmath,amssymb,amsthm}
\usepackage{mathrsfs,url,hyperref,tikz}

\title{Interior-Boundary Conditions for the Dirac Equation at Point Sources in 3 Dimensions}

\author{
Joscha Henheik\footnote{Mathematisches Institut,
     Eberhard-Karls-Universit\"at, Auf der Morgenstelle 10, 72076
     T\"ubingen, Germany and Institute of Science and Technology Austria (IST Austria), Am Campus 1, 3400 Klosterneuburg, Austria; E-mail:
     joscha.henheik@ist.ac.at}~~and
Roderich Tumulka\footnote{Mathematisches Institut,
	Eberhard-Karls-Universit\"at, Auf der Morgenstelle 10, 72076
	T\"ubingen, Germany; E-mail:
     roderich.tumulka@uni-tuebingen.de}
}
\date{June 20, 2022}

\newcommand{\Hilbert}{\mathscr{H}}
\newcommand{\Kilbert}{\mathscr{K}}
\newcommand{\sA}{\mathscr{A}}

\newcommand{\Q}{\mathcal{Q}}

\renewcommand{\Im}{\mathrm{Im}}

\newcommand{\RRR}{\mathbb{R}}
\newcommand{\CCC}{\mathbb{C}}
\newcommand{\SSS}{\mathbb{S}}

\newcommand{\nr}{{\mathrm{nr}}} 
\newcommand{\free}{{\mathrm{free}}}
\newcommand{\Nmax}{{N_{\max}}}
\newcommand{\ve}{\boldsymbol{e}}

\newcommand{\vx}{\boldsymbol{x}}

\newcommand{\vJ}{\boldsymbol{J}}
\newcommand{\vL}{\boldsymbol{L}}

\newcommand{\vS}{\boldsymbol{S}}
\newcommand{\valpha}{\boldsymbol{\alpha}}

\newcommand{\vomega}{\boldsymbol{\omega}}
\newcommand{\vzero}{\boldsymbol{0}}

\newtheorem{thm}{Theorem}
\newtheorem{lem}{Lemma}
\theoremstyle{definition}\newtheorem{rmk}{Remark}
\newcommand{\be}{\begin{equation}}
\newcommand{\ee}{\end{equation}}

\begin{document}
\maketitle
\begin{abstract}
A recently proposed approach for avoiding the ultraviolet divergence of Hamiltonians with particle creation is based on interior-boundary conditions (IBCs). The approach works well in the non-relativistic case, that is, for the Laplacian operator. Here, we study how the approach can be applied to Dirac operators. While this has been done successfully already in 1 space dimension, and more generally for codimension-1 boundaries, the situation of point sources in 3 dimensions corresponds to a codimension-3 boundary. One would expect that, for such a boundary, Dirac operators do not allow for boundary conditions because they are known not to allow for point interactions in 3d, which also correspond to a boundary condition. And indeed, we confirm this expectation here by proving that there is no self-adjoint operator on (a truncated) Fock space that would correspond to a Dirac operator with an IBC at configurations with a particle at the origin. However, we also present a positive result showing that there are self-adjoint operators with IBC (on the boundary consisting of configurations with a particle at the origin) that are, away from those configurations, given by a Dirac operator plus a sufficiently strong Coulomb potential.

\medskip

  \noindent 
PACS: 
11.10.Ef; 	
03.70.+k; 	
11.10.Gh. 	
  Key words: 
  self-adjoint extension of Dirac operator;
  particle creation;
  ultraviolet infinity;
  regularization of quantum field theory;
  probability current.
\end{abstract}

\newpage

\section{Introduction}

Hamiltonians for quantum theories with particle creation and annihilation are often plagued by ultraviolet divergence \cite{Opp30,vH52,Der03}. For defining such a Hamiltonian in a rigorous way, one might employ an ultraviolet cut-off corresponding to smearing out the source of particle creation over a positive volume, or in some cases one can obtain a renormalized Hamiltonian by taking a limit of the cut-off Hamiltonian in which the volume of the source tends to zero \cite{Lee54,Nel64,GJ85,GJ87}. Another, more recent approach is based on interior-boundary conditions (IBCs) \cite{TT15a,TT15b} and yields directly (i.e., without taking a limit) a Hamiltonian suitable for a point source. Here, the wave function $\psi$ is a function on the configuration space 
\be\label{Qdef}
\Q=\bigcup_{N=0}^\infty \Q_1^N
\ee
of a variable number of particles that can move in the 1-particle space $\Q_1$, for example
\be\label{Q1def}
\Q_1=\RRR^{3}\setminus\{\vzero\}
\ee
if particles can be created at a point source fixed at the origin $\vzero\in\RRR^3$. In this example, the boundary $\partial\Q$ of $\Q$ consists of configurations with a particle at the origin,
\be\label{dQdef}
\partial\Q = \bigcup_{N=0}^\infty \partial (\Q_1^N) = \bigcup_{N=0}^\infty \Bigl\{(\vx_1,\ldots,\vx_N)\in(\RRR^3)^N: \vx_i=\vzero \text{ for some }i\in\{1,\ldots,N\} \Bigr\}.
\ee
The IBC is a condition on $\psi$ relating values on the boundary and values in the interior, more precisely, relating the values on two configurations $q,q'$ that differ by the creation (respectively, annihilation) of a particle at the origin,
\begin{subequations}\label{qq'}
\begin{align}
q&=(\vx_1,\ldots,\vx_i=\vzero,\ldots,\vx_N)\in \partial\Q_1^N\text{ and}\\
q'&=(\vx_1,\ldots,\vx_{i-1},\vx_{i+1},\ldots,\vx_N)\in \Q_1^{N-1}\,.
\end{align}
\end{subequations}
As shown by Lampart et al.~\cite{ibc2a}, a Hamiltonian with particle creation at the origin can be defined rigorously in this way in the non-relativistic case without UV divergence or the need for renormalization. Our goal here is to examine in what way and to what extent this approach can be extended to the Dirac equation. 

There are already two works about IBCs for the Dirac equation: Schmidt et al.~\cite{IBCdiracCo1} have explored what IBCs for the Dirac equation can look like on a codimension-1 boundary. However, in our case the boundary \eqref{dQdef} has codimension 3. Lienert and Nickel~\cite{LN18} developed a quantum field theory (QFT) model in 1 space dimension using Dirac particles and an IBC that allows two particles to merge into one or one particle to split in two. In contrast, we consider here space dimension 3.

We present here a negative result and a positive one. The negative result (Theorem~\ref{thm:nogo}) asserts, roughly speaking, that in 3 space dimensions, there exists no self-adjoint Hamiltonian for the configuration space as in \eqref{Qdef}--\eqref{dQdef} (or even the truncated one allowing only $N=0$ and $N=1$) that acts like the free Dirac Hamiltonian away from the boundary but involves particle creation. Put differently, the free Dirac equation cannot be combined with IBCs in 3 dimensions. The fact is analogous to the known impossibility of point interaction ($\delta$ potentials) for the Dirac equation in 3 dimensions (see, e.g., \cite{Sve81,AGHKH88} and \cite[Theorem 1.1]{Tha91}), in particular since point interaction is described by a boundary condition \cite{BP35}, and for an external field that acts nontrivially only at $\vzero\in\RRR^3$, the relevant boundary for this boundary condition is precisely \eqref{dQdef}, or one sector thereof. Our proof of Theorem~\ref{thm:nogo} makes use of the impossibility theorems for point interaction to deduce the impossibility of IBCs in this situation. We formulate some variants of Theorem~\ref{thm:nogo} as Theorems~\ref{thm:general} and \ref{thm:weaknogo}.

The positive result (Theorem~\ref{thm:sa}) concerns a way in which nevertheless a self-adjoint Hamiltonian with particle creation can be rigorously defined by means of an IBC and the Dirac equation for the configuration space as in \eqref{Qdef}--\eqref{dQdef}; it is based on adding a potential. Specifically, we show in Theorem~\ref{thm:sa} that if a Coulomb potential of sufficient strength, centered at the origin and acting on each of the particles in the model, is added to the action of the Hamiltonian away from the boundary, then a self-adjoint version of the Hamiltonian exists that involves an IBC and leads to particle creation and annihilation at the origin, analogously to the non-relativistic case with the Laplacian operator. The IBC is analogous to the known IBCs for codimension-1 boundaries \cite{IBCdiracCo1}. We formulate the result for a truncated Fock space with only the $N=0$ and $N=1$ sectors. 

The results are formulated in detail in Section~\ref{sec:results}. They can be expressed as statements about self-adjoint extensions. That is because, when considering a truncated Fock space ($N\leq N_{\max}$) and configuration space then, for functions $\psi$ that vanish in a neighborhood of the boundary in the top sector $N_{\max}$ and also vanish in all lower sectors, we know how the desired Hamiltonian should act: like the free Dirac operator $H^\free$ of $N_{\max}$ particles (respectively, with a Coulomb potential). Let $D^\circ$ be the space of these functions (not dense). Thus, the desired Hamiltonian $H$ is a self-adjoint extension of a (symmetric but not closed) operator $H^\circ= H|_{D^\circ}=H^\free|_{D^\circ}$. (By the way, when we say ``self-adjoint,'' we always mean that the operator is densely defined.) The no-go result (Theorem~\ref{thm:nogo}) will show that $H^\circ$ has only one self-adjoint extension, viz., $H^\free$; that means that it is not possible to implement particle creation and annihilation at just one point $\vzero$, using IBCs or otherwise. Theorem~\ref{thm:sa} will take $H^\circ$ to include a suitable Coulomb potential and provide a self-adjoint extension (even several ones) featuring particle creation and annihilation. It remains to be seen whether and how IBCs can be employed in more realistic models of relativistic QFTs.

The IBC Hamiltonians provided by Theorem~\ref{thm:sa} are neither translation invariant nor rotation invariant, but that was only to be expected: the model cannot be translation invariant, given that the source is fixed at the origin, and the emission of a spin-$\tfrac12$ particle by a spinless source cannot conserve angular momentum and thus cannot be rotation invariant. (An alternative proof is given in Section~\ref{sec:rotation}.)\footnote{More generally, the emission of a spin-$\tfrac12$ particle by a source of \emph{any} spin cannot conserve angular momentum because (i)~the creation term in the Hamiltonian morally reads $\sum_{r,r',s} g_{rr's} a_r^\dagger(\vx) \, b_s^\dagger(\vx) \, a_{r'}(\vx)$ with $a_r^\dagger$ the creation operator of the emitting particle, $b_s^\dagger$ the creation operator of the emitted particle, and $g_{rr's}$ complex coefficients with $s$ a Dirac spinor index and $r$ and $r'$ indices referring to another representation space of the rotation group $SO(3)$ or its covering group; and (ii)~there is no rotation invariant array of coefficients $g_{rr's}$. In contrast, the emission of a spin-1 particle by a source of spin $\tfrac12$ \emph{can} conserve angular momentum, as the coefficients $g_{ss'\mu}=\gamma_{ss'\mu}$ are rotation invariant.}

Here are further comments on the literature. Some forms of interior-boundary conditions were considered, not necessarily with the UV problem in mind, early on by Moshinsky \cite{Mosh51a,Mosh51b,Mosh51c}. Self-adjoint Hamiltonians based on IBCs and the Laplacian were first rigorously defined by Thomas \cite{Tho84} and Yafaev \cite{Yaf92}. Lampart et al.~\cite{ibc2a} extended these results to the full Fock space and showed that the non-relativistic IBC Hamiltonian agrees, up to addition of a constant, with the one obtained through UV cut-off and renormalization. Lampart and Schmidt \cite{LS18} further extended the proofs to moving sources in 2 space dimensions and Lampart \cite{Lam18} in 3 dimensions. Keppeler and Sieber \cite{KS15} studied the 1-dimensional non-relativistic case. Bohmian trajectories for IBC Hamiltonians were defined and studied in \cite{Tum04,bohmibc}. Further studies of IBC Hamiltonians and their properties include \cite{ML91,timelike,Gal16,Sch18,ST18,co1}.

In Section~\ref{sec:results}, we describe our results. In Section~\ref{sec:proofs}, we collect the proofs.

\section{Results}
\label{sec:results}

For comparison, it will be useful to recapitulate some aspects of the non-relativistic case in 3 dimensions with a single point source fixed at the origin \cite{ibc2a}.  For simplicity, we consider a truncated Fock space of spinless particles,
\be
\Hilbert_\nr := \bigoplus_{N=0}^\Nmax \Hilbert^{(N)}_\nr := \bigoplus_{N=0}^\Nmax S_\pm L^2(\RRR^3,\CCC)^{\otimes N}
\ee
with $\oplus$ the orthogonal sum of Hilbert spaces and $S_\pm \cdots$ the image of the symmetrization operator $S_+$ or anti-symmetrization operator $S_-$ (and subscript $\nr$ for ``non-relativistic case''). 
There is a 5-parameter family of IBC Hamiltonians; some members of this family can be regarded as involving an external zero-range potential at the origin in addition to the particle source; there is a 2-parameter subfamily that can be regarded as involving no such potential, i.e., as being a pure particle source. The remaining parameters are the energy $E_0$ that must be expended for creating a particle and the strength $g$ of particle creation. Let us fix values $E_0>0$ and $g\in\CCC\setminus\{0\}$ and call the corresponding operator $H_\nr$. $H_\nr$ is a self-adjoint operator in $\Hilbert_\nr$; let $D_\nr$ denote its domain of self-adjointness. Functions $\psi=(\psi^{(0)},\ldots,\psi^{(\Nmax)})$ in $D_\nr$ satisfy the IBC
\be\label{IBC1}
\lim_{r\searrow 0} r \psi^{(N+1)}(\vx_1,\ldots,\vx_N,r\vomega) = -\frac{gm}{2\pi\hbar^2\sqrt{N+1}} \psi^{(N)}(\vx_1,\ldots,\vx_N)
\ee
(suitably understood, and with $r\searrow 0$ meaning the limit from the right) for $N=0,\ldots,\Nmax-1$ and every unit vector $\vomega\in\RRR^3$. Another operator to compare to is the free Hamiltonian $H^\free_\nr$, which acts on functions $\psi$ from its domain
\be
D^\free_\nr := \bigoplus_{N=0}^\Nmax S_\pm H^2(\RRR^{3N},\CCC)
\ee
(with $H^2$ the second Sobolev space) according to
\be
H^\free_\nr \psi^{(N)} = \sum_{j=1}^N H_{1\nr,j} \psi^{(N)}
\ee
with $H_{1\nr,j}$ the 1-particle Hamiltonian
\be
H_{1\nr} = E_0-\tfrac{\hbar^2}{2m}\Delta
\ee
acting on particle $j$.

Now we want to express that away from the boundary, $H_\nr$ acts like $H^\free_\nr$. However, the IBC enforces that if $\psi^{(N)}\neq 0$ for some $N\geq 0$, then every higher sector, $\psi^{(N')}$ with $N'>N$, must be nonzero (and, in fact, unbounded) in a neighborhood of the boundary $\partial\Q_1^{N'}$. But $\psi$ can lie in $D_\nr$ if $\psi^{(\Nmax)}$ vanishes in a neighborhood of the boundary and $\psi^{(N)}=0$ for all $N<\Nmax$; in fact, $\psi$ \emph{will} lie in $D_\nr$ if $\psi^{(\Nmax)}\in C_c^\infty(\Q_1^\Nmax,\CCC)$ (where $C_c^\infty$ means smooth functions with compact support) and $\psi^{(N)}=0$ for all $N<\Nmax$,
\be
D_\nr \supset D^\circ_\nr :=\{0\}\oplus \ldots \oplus\{0\}\oplus C_c^\infty(\Q_1^\Nmax,\CCC)\,.
\ee
Note that since $\vzero$ was excluded from $\Q_1$, ``compact support'' entails that the support stays away from the boundary. The symbol $\oplus$, when applied to sets that are not Hilbert spaces, should be understood as the Cartesian product, resulting in a subset of $\Hilbert_\nr$.
On $D^\circ_\nr$, $H_\nr$ acts like the free Hamiltonian,
\be
H_\nr \Big|_{D^\circ_\nr} = H^\free_\nr \Big|_{D^\circ_\nr} =: H^\circ_\nr\,. 
\ee
That is, both $(H_\nr,D_\nr)$ and $(H^\free_\nr,D^\free_\nr)$ are self-adjoint extensions of $(H^\circ_\nr,D^\circ_\nr)$. Note that $D^\circ_\nr$ is not a \emph{dense} subspace of $\Hilbert_\nr$ (whereas $D_\nr$ and $D^\free_\nr$ are). 
The condition that $(H_\nr,D_\nr)$ is an extension of $(H^\circ_\nr, D^\circ_\nr)$ expresses that in the highest sector, particle creation can occur only at the origin.
In passing, we remark that Yafaev \cite{Yaf92} showed for $\Nmax=1$ that all self-adjoint extensions of $( H^\circ_\nr, D^\circ_\nr)$ belong to the 5-parameter family of IBC Hamiltonians (which includes, as a subfamily, the Hamilonians without particle creation but with point interaction at the origin).

\subsection{No-Go Theorem}

Now we turn to the Dirac case. We define the truncated Fock space
\be
\Hilbert := \bigoplus_{N=0}^\Nmax \Hilbert^{(N)}:=\bigoplus_{N=0}^\Nmax S_\pm L^2(\RRR^3,\CCC^4)\,.
\ee
(Although spin-$\frac12$ particles are fermions in nature, we cover here also the mathematical case of Dirac particles that are bosons.) The 1-particle Hamiltonian is the Dirac Hamiltonian
\be
H_1 = -ic\hbar \valpha\cdot \nabla + mc^2\beta
\ee
with mass $m\geq 0$, which is self-adjoint on $D_1=H^1(\RRR^3,\CCC^4)$ (first Sobolev space). Let $H^\free_N$ be the free $N$-particle Hamiltonian, which acts according to
\be
H^\free_N \psi = \sum_{j=1}^N H_{1j} \, \psi\,,
\ee
$D(H^\free_N)$ its domain of self-adjointness in $L^2(\RRR^{3N},(\CCC^4)^{\otimes N})$, and
\be
H^\free \psi^{(N)} = H^\free_N \psi^{(N)}
\ee
on the domain
\be
D^\free = \bigoplus_{N=0}^\Nmax S_\pm D(H^\free_N) \,.
\ee
$H^\free$ is self-adjoint, and is the (truncated) ``second quantization'' of $H_1$ in $\Hilbert$.

The desired IBC Hamiltonian $(H,D)$, or in fact any Hamiltonian that agrees with $H^\free$ except for particle creation and annihilation at the origin, must be a self-adjoint extension of $(H^\circ,D^\circ)$ with
\be
D^\circ = \{0\}\oplus \ldots \oplus \{0\} \oplus S_{\pm}
C_c^\infty \bigl(\Q_1^\Nmax,(\CCC^4)^{\otimes N_{\max}}\bigr)
\ee
and
\be
H^\circ := H^\free \Big|_{D^\circ}\,.
\ee
Let $\Hilbert^{(<\Nmax)}:=\Hilbert^{(0)}\oplus \ldots \oplus \Hilbert^{(\Nmax-1)}$.

\begin{thm}\label{thm:nogo}
Let $\Nmax>0$.
For every self-adjoint extension $(H,D)$ of $(H^\circ,D^\circ)$, the highest sector decouples from the other sectors; that is, $H$ is block diagonal with respect to the decomposition $\Hilbert=\Hilbert^{(<\Nmax)}\oplus \Hilbert^{(\Nmax)}$.
\end{thm}

We give all proofs in Section~\ref{sec:proofs}. To paraphrase the conclusion of Theorem~\ref{thm:nogo}, the time evolution generated by $H$ involves no exchange between $\psi^{(\Nmax)}$ and the other sectors; no particle creation or annihilation occurs towards or from the $N_{\max}$-sector; in particular, $\|\psi^{(\Nmax)}\|$ is time independent. So Theorem~\ref{thm:nogo} implies that there is no IBC Hamiltonian for the Dirac equation in 3 space dimensions, as long as no further element such as potentials, space-time curvature, or other particles is introduced. 

Theorem~\ref{thm:nogo} is obtained by combining two theorems, Theorem~\ref{thm:Sve} and Theorem~\ref{thm:general}. The former is a specialized form of a theorem of Svendsen \cite{Sve81}:

\begin{thm}\label{thm:Sve}
Let $N\geq 1$ and $M\subset \RRR^{3N}$ the union of finitely many $C^\infty$ submanifolds of equal codimension $c$. Let 
$H_{\setminus M}$ be the restriction of $H^\free_N$ to $C_c^\infty(\RRR^{3N}\setminus M,(\CCC^4)^{\otimes N})$. Then $H_{\setminus M}$ is essentially self-adjoint in $L^2(\RRR^{3N},(\CCC^4)^{\otimes N})$ if and only if $c\geq 2$. In particular, the restriction of $H^\free_N$ to $C_c^\infty(\Q_1^N,(\CCC^4)^{\otimes N})$ is essentially self-adjoint, and its restriction to $S_{\pm}C_c^\infty(\Q_1^N,(\CCC^4)^{\otimes N})$ is essentially self-adjoint in $S_{\pm}L^2(\RRR^{3N},(\CCC^4)^{\otimes N})$.
\end{thm}

The last sentence follows by taking
\begin{subequations}
\begin{align}
M&=\bigl\{(\vx_1,\ldots,\vx_N)\in\RRR^{3N}: \vx_i=\vzero \text{ for some }i\bigr\}\\
&= \bigcup_{i=1}^N \bigl\{(\vx_1,\ldots,\vx_N)\in\RRR^{3N}: \vx_i=\vzero\bigr\},
\end{align}
\end{subequations}
so $c=3$ and $\RRR^{3N}\setminus M = \Q_1^N$.
This theorem excludes point interaction for the free Dirac equation in 3 space dimensions.
Note, however, that Theorem~\ref{thm:nogo} is not a direct corollary of Theorem~\ref{thm:Sve} because our Hilbert space $\Hilbert$ is not $L^2(\RRR^{3N},\CCC^k)$ but contains further sectors, and $D^\circ$ is not dense. What we need is the following statement, a kind of generalization of Theorem~\ref{thm:nogo}:

\begin{thm}\label{thm:general}
Let $\Nmax>0$, let $(\tilde H,\tilde D)$ be essentially self-adjoint in $\Hilbert^{(\Nmax)}$, and let now
\be
D^\circ:=\{0\}\oplus \ldots \oplus \{0\}\oplus \tilde D \subset \Hilbert
\ee
and $H^\circ: D^\circ\to \Hilbert$ be given by
\be
H^\circ \bigl(0,\ldots,0,\psi^{(\Nmax)}\bigr) := \bigl(0,\ldots,0,\tilde H \psi^{(\Nmax)}\bigr)\,.
\ee
For every self-adjoint extension $(H,D)$ of $(H^\circ,D^\circ)$, the highest sector decouples from the other sectors; in fact, with respect to the decomposition $\Hilbert=\Hilbert^{(<\Nmax)}\oplus \Hilbert^{(\Nmax)}$, $D=D^{(<\Nmax)}\oplus D^{(\Nmax)}$ and $H$ is block diagonal with blocks $H^{(<\Nmax)}$ and $H^{(\Nmax)}$, where $(H^{(\Nmax)},D^{(\Nmax)})$ is the unique self-adjoint extension of $(\tilde H, \tilde D)$.
\end{thm}

\subsection{Hamiltonians with Coulomb Potential}

In this section and Section~\ref{sec:exist}, we focus on the case of two sectors, i.e., $\Nmax=1$.
Our positive result is about examples of Dirac Hamiltonians in 3d with particle creation by means of IBCs. These Hamiltonians are based on the 1-particle Hamiltonian
\be\label{H1Coulomb}
H_1 = -ic\hbar \valpha\cdot \nabla + mc^2\beta + \frac{q}{|\vx|}
\ee
that consists of the free Dirac Hamiltonian plus a Coulomb potential of strength $q$ (i.e., $q$ is the product of the charge at $\vx$ and the charge at the origin). We will show in Theorem~\ref{thm:sa} that for $\sqrt{3}/2<|q|<1$, there exist IBC Hamiltonians. We conjecture that also for $|q| \geq 1$, IBC Hamiltonians exist. On the other hand, the following theorem, a generalization of Theorem~\ref{thm:nogo} in the case $\Nmax=1$, shows that for $|q|\leq \sqrt{3}/2$, no IBC Hamiltonian exists:

\begin{thm}\label{thm:weaknogo}
Let $\Hilbert= \CCC \oplus L^2(\RRR^3,\CCC^4)$, $D^\circ=\{0\}\oplus C_c^\infty(\Q_1,\CCC^4)$, and let $H^\circ$ act on $\psi=(0,\psi^{(1)})\in D^\circ$ like $H_1$ as in \eqref{H1Coulomb},
\be\label{Hcircdef}
H^\circ\psi = \bigl(0,H_1 \psi^{(1)}\bigr)
\ee
with $|q|\leq \sqrt{3}/2$. For every self-adjoint extension $(H,D)$ of $(H^\circ,D^\circ)$, the two sectors decouple, that is, $H$ is block diagonal. In fact, $D=\CCC \oplus D\bigl(\overline{H_1}\bigr)$ with $\overline{H_1}$ the closure of $H_1$, and for $\psi\in D$,
\be\label{Hweaknogo}
H\psi = \bigl(E_{00}\psi^{(0)}, \overline{H_1} \psi^{(1)}\bigr)
\ee
with some constant $E_{00}\in\RRR$.
\end{thm}

It is known (e.g., \cite[Prop.~A1]{ELS19}) that for $|q|<\sqrt{3}/2$, 
$D\bigl(\overline{H_1}\bigr) = H^1(\RRR^3,\CCC^4)$ (first Sobolev space), whereas for $|q|=\sqrt{3}/2$, the domain is bigger than the first Sobolev space. Theorem~\ref{thm:weaknogo} follows by means of Theorem~\ref{thm:general} from the following known theorem \cite[Thm.~6.9]{Wei87}, \cite{Gal17}: 

\begin{thm}\label{thm:weakCoulomb}
In $L^2(\RRR^3,\CCC^4)$, the operator $H_1$ as in \eqref{H1Coulomb} is essentially self-adjoint on $C_c^\infty(\Q_1,\CCC^4)$ if and only if $|q|\leq \sqrt{3}/2$.
\end{thm}

\begin{rmk}
It is also known \cite{Tha91} that 
\be\label{H1matrix}
H_1=-ic\hbar \valpha\cdot \nabla + mc^2\beta + V(\vx)
\ee
with a matrix-valued potential $V$ such that each component $V_{ij}$ satisfies the bound
\be\label{Vmatrix}
|V_{ij}(\vx)| \leq \frac{q}{|\vx|} +b
\ee
with constants $b>0$ and $0<q\leq 1/2$ is essentially self-adjoint. By Theorem~\ref{thm:general}, Theorem~\ref{thm:weaknogo} still applies if \eqref{H1Coulomb} is replaced by \eqref{H1matrix} and \eqref{Vmatrix} with $0<q\leq 1/2$.
\end{rmk}

\subsection{Existence of IBC Hamiltonian}
\label{sec:exist}

It will be helpful again to consider first the non-relativistic IBC Hamiltonian $(H_\nr,D_\nr)$, now for $\Nmax=1$, so $\Hilbert_\nr=\CCC\oplus L^2(\RRR^3,\CCC)$. We report a few facts \cite{ibc2a}: For every $\psi\in D_\nr$, the upper sector is of the form
\be\label{shortnr}
\psi^{(1)}(\vx) = c_{-1} |\vx|^{-1} + c_0 |\vx|^0 + o(|\vx|^0)
\ee
as $\vx\to \vzero$, with (uniquely defined) ``short distance coefficients'' $c_{-1},c_0\in \CCC$. $\psi\in D_\nr$ satisfies the IBC \eqref{IBC1}, which can be written in the form
\be
c_{-1} = g\,  \psi^{(0)}\,,
\ee
and the Hamiltonian acts on $\psi$ like 
\begin{subequations}
\begin{align}
(H_\nr\psi)^{(0)} &= g^* \, c_0\,,\\
(H_\nr\psi)^{(1)}(\vx) &= (E_0 -\tfrac{\hbar^2}{2m}\Delta)\psi^{(1)}(\vx)~~~\text{for }\vx\neq \vzero\,.
\end{align}
\end{subequations}
(At $\vx=\vzero$, the $\Delta\psi^{(1)}$, when understood in the distributional sense, includes a delta distribution stemming from the $|\vx|^{-1}$ contribution in \eqref{shortnr}.)

Now we turn again to the Dirac case with 
\be\label{miniFock}
\Hilbert=\CCC\oplus L^2(\RRR^3,\CCC^4)
\ee
and 1-particle operator $H_1$ as in \eqref{H1Coulomb}.
In the following, we take $H_1$ to be an operator on the domain 
\be\label{D1def}
D_1=C_c^\infty(\Q_1,\CCC^4)\subset L^2(\RRR^3,\CCC^4)
\ee
with adjoint $(H_1^*,D_1^*)$ in $L^2(\RRR^3,\CCC^4)$. 

The notation
\be
\Phi^{\pm}_{m_j,\kappa_j}
\ee
is common for certain functions that form an orthonormal basis of $L^2(\SSS^2,\CCC^4)$ (where $\SSS^2$ means the unit sphere in $\RRR^3$) and are simultaneous eigenvectors of $\vJ^2,K,J_3$ with $\vJ=\vL+\vS$ the total angular momentum and $K=\beta(2\vS\cdot \vL+1)$ the ``spin-orbit operator.'' Their explicit definition in terms of spherical harmonics can be found in, e.g., \cite[p.~126]{Tha91}. The symbols $m_j$ and $\kappa_j$ are the traditional names of their indices.

\begin{thm}\label{thm:sa}
Let $\Hilbert= \CCC \oplus L^2(\RRR^3,\CCC^4)$, $D^\circ=\{0\}\oplus C_c^\infty(\Q_1,\CCC^4)$, and $H^\circ$ be given by \eqref{Hcircdef} with $\sqrt{3}/2<|q|<1$. Set $B:=\sqrt{1-q^2}$ and note that $0<B<\tfrac12$. Choose $g\in\CCC\setminus\{0\}$ and
\be\label{mjkappaj}
(\tilde m_j,\tilde\kappa_j) \in \sA:=\Bigl\{ (-\tfrac12,-1),~(-\tfrac12,1),~(\tfrac12,-1),~(\tfrac12,1) \Bigr\}\,.
\ee
Then there is a self-adjoint extension $(H,D)$ of $(H^\circ,D^\circ)$ with the following properties:
\begin{enumerate}
\item Particle creation occurs, i.e., the two sectors do not decouple ($H$ is not block diagonal in the decomposition \eqref{miniFock}).
\item For every $\psi\in D$, the upper sector is of the form
\be\label{short}
\psi^{(1)}(\vx) = c_-\, f^{-}_{\tilde m_j \tilde\kappa_j}\bigl(\tfrac{\vx}{|\vx|}\bigr)\, |\vx|^{-1-B} + \hspace{-3mm}\sum_{(m_j,\kappa_j)\in\sA} \hspace{-3mm} c_{+ m_j \kappa_j}\, f^+_{m_j \kappa_j}\bigl(\tfrac{\vx}{|\vx|}\bigr) \, |\vx|^{-1+B} + o(|\vx|^{-1/2})
\ee
as $\vx\to \vzero$ with (uniquely defined) short distance coefficients $c_-,c_{+ m_j \kappa_j}\in \CCC$ and particular functions $f^{\pm}_{m_j \kappa_j}:\SSS^2\to\CCC^4$ given by
\begin{subequations}\label{fdef}
\begin{align}
f^+_{m_j \kappa_j} &= (1+q-B) \Phi^+_{m_j\kappa_j}-(1+q+B)\Phi^-_{m_j\kappa_j}\label{f+def}\\
f^{-}_{m_j \kappa_j} &= (1+q+B)\Phi^+_{m_j\kappa_j}- (1+q-B) \Phi^-_{m_j\kappa_j}\,.\label{f-def}
\end{align}
\end{subequations}
\item Every $\psi\in D$ obeys the IBC
\be\label{IBC}
c_- = g \, \psi^{(0)}\,,
\ee
and $H$ acts on $\psi\in D$ according to
\begin{subequations}\label{Hact}
\begin{align}
(H\psi)^{(0)}&= g^* \,4B(1+q)\, c_{+ \tilde m_j \tilde\kappa_j}\label{Hact0}\\
(H\psi)^{(1)}(\vx)&= \Bigl(-ic\hbar \valpha\cdot \nabla + mc^2\beta + \frac{q}{|\vx|}\Bigr)\psi^{(1)}(\vx)~~~~~(\vx\neq \vzero)\,.\label{Hact1}
\end{align}
\end{subequations}
\end{enumerate}
\end{thm}

Equivalently, the IBC \eqref{IBC} could be written as 
\be\label{IBC2}
\lim_{r\searrow 0} r^{1+B}\psi^{(1)}(r\vomega) = g\, f^-_{\tilde m_j \tilde\kappa_j}(\vomega) \, \psi^{(0)}
\ee
for every $\vomega\in\SSS^2$.

\begin{rmk}
It seems plausible that an analogous IBC Hamiltonian can be set up for $\Nmax>1$ by applying the same terms on each particle sector.
\end{rmk}

\begin{rmk}\label{rmk:family}
The IBC Hamiltonian described in Theorem~\ref{thm:sa} belongs to a whole family of IBC Hamiltonians in which \eqref{IBC} and \eqref{Hact0} are replaced by
\be\label{IBCgen}
a_1 \, c_- + a_2 \, c_{+ \tilde m_j \tilde\kappa_j} = g \, \psi^{(0)}
\ee
and
\be\label{Hact0gen}
(H\psi)^{(0)}= g^* \,(a_3\, c_- + a_4 \, c_{+ \tilde m_j \tilde\kappa_j})\,,
\ee
respectively, with real constants $a_1,\ldots,a_4$ satisfying
\be\label{acond}
a_1 a_4-a_2 a_3 = 4B(1+q)\,.
\ee
As we will show in the proof of Theorem~\ref{thm:sa}, each of these choices defines a self-adjoint operator. The Hamiltonian defined through \eqref{IBC} and \eqref{Hact0} corresponds to $a_1=1$, $a_2=0$, $a_3=0$, and $a_4=4B(1+q)$.
\end{rmk}

\begin{rmk}
In our construction of $H$, one angular momentum sector of $L^2(\RRR^3,\CCC^4)$, indexed by $(\tilde m_j,\tilde\kappa_j)$, gets coupled to the 0-particle sector. If we took the 0-particle sector to have more than 1 dimension, say $\Hilbert^{(0)}=\CCC^4$ instead of $\CCC$, then we could couple several angular momentum sectors listed in \eqref{mjkappaj} to several mutually orthogonal subspaces of the 0-particle sector $\Hilbert^{(0)}$. However, this mathematical possibility does not seem physically natural.
\end{rmk}

\begin{rmk}
It would be interesting to investigate also the case $|q|\geq 1$. We expect that IBC Hamiltonians exist also in that case because $h_{m_j\kappa_j}$ (an angular momentum block of $H_1$, see Section~\ref{sec:proofs}) is known to have multiple self-adjoint extensions also in that case. However, we do not have a proof.
\end{rmk}

\begin{rmk}
It is also of interest to define a $|\psi|^2$-distributed jump process for the Bohmian configuration in analogy to the processes defined in \cite{bohmibc} for non-relativistic IBC Hamiltonians. We plan to address this issue in a separate work.
\end{rmk}

\noindent\underline{\textit{Note added:}} After completion of this work, Binz and Lampart posted a preprint \cite{BL21} in which they described a general abstract framework for the construction of IBC Hamiltonians. The one given in Theorem \ref{thm:sa} can also be cast in this form. 

\subsection{Rotational Symmetry}
\label{sec:rotation}

The Hamiltonian $H$ provided by Theorem~\ref{thm:sa} is not rotationally symmetric, for example because the subspace $\Kilbert_{\tilde m_j \tilde\kappa_j}$ that plays a special role for $H$ is not invariant under $J_1$ or $J_2$, and thus not under the action of the rotation group (more precisely, of its covering group). One might think of coupling all four angular momentum sectors $\Kilbert_{m_j\kappa_j}$ with $(m_j,\kappa_j)\in\sA$ to the 0-particle sector in a symmetric way, but actually that does not help:

\begin{thm}\label{thm:rotation}
Let $\Hilbert$ and $(H^\circ,D^\circ)$ be as in Theorems \ref{thm:weaknogo} and \ref{thm:sa} with any $q\in \RRR$.
None of the self-adjoint extensions of $(H^\circ,D^\circ)$ with particle creation is rotationally symmetric.
\end{thm}

The proof is based on the following fact that we also prove in Section~\ref{sec:proofs}.

\begin{lem}\label{lem:rotation}
The only vector in $L^2(\RRR^3,\CCC^4)$ invariant under rotations (i.e., under the representation of the covering group of $SO(3)$) is the zero vector.
\end{lem}

\section{Proofs}
\label{sec:proofs}

\begin{proof}[Proof of Theorem~\ref{thm:nogo}]
Let $\tilde{D}=S_{\pm} C_c^\infty\bigl(\Q_1^{N_{\max}},(\CCC^4)^{\otimes N_{\max}}\bigr)$ and $\tilde H$ the restriction of $H^\free_N$ to $\tilde D$. By Theorem~\ref{thm:Sve}, $(\tilde H,\tilde D)$ is essentially self-adjoint in $\Hilbert^{(N_{\max})}$.
Now Theorem~\ref{thm:general} yields the statement of Theorem~\ref{thm:nogo}.
\end{proof}

\begin{proof}[Proof of Theorem~\ref{thm:general}]
As a shorthand notation, we write
\be
\Hilbert^{<}:= \Hilbert^{(<N_{\max})}\,,~~~\Hilbert^{=} := \Hilbert^{(N_{\max})}
\ee
and correspondingly $H^{<},D^{<},H^{=},D^{=}$. 
Let $\Gamma(A)$ denote the graph of an operator $(A,D(A))$ and $\overline{\Gamma(A)}$ its closure in $\Hilbert\oplus \Hilbert$, which is the graph of the closure of $A$, $\Gamma\bigl(\overline{A}\bigr)$. 
Since $\tilde H$ is essentially self-adjoint, its closure $\overline{\tilde H}$ is its self-adjoint extension $H^{=}$.
Since $(H^\circ, D^\circ)$ is symmetric (and 
$\tilde H$ is densely defined), it is closable with closure $\bigl( \overline{H^\circ}, D\bigl(\overline{H^\circ}\bigr) \bigr)$, and we get that
\begin{subequations}\label{graphH}
\begin{align}
	 \Gamma(H)  = \overline{\Gamma(H)} 
	 &\supset\overline{\Gamma(H^\circ)}
	 = \overline{\left\{\begin{pmatrix}
		0 \\ f \\ 0 \\ \tilde H f \end{pmatrix}  
	: f \in \tilde D \right\}}
	= \left\{\begin{pmatrix}
		0 \\ f \\ 0 \\ g \end{pmatrix}  
	: \begin{pmatrix}f\\g\end{pmatrix} \in \overline{\Gamma(\tilde H)} \right\} \\[3mm]
	&=\left\{\begin{pmatrix}
		0 \\ f \\ 0 \\ g \end{pmatrix}  
	: \begin{pmatrix}f\\g\end{pmatrix} \in \Gamma(H^{=}) \right\}
	=\left\{\begin{pmatrix}
		0 \\ f \\ 0 \\ H^{=}f \end{pmatrix}  
	: f \in D^{=} \right\}\,,
\end{align}
\end{subequations}
where $0$ means the zero of $\Hilbert^{<}$. In particular,
\begin{equation}\label{domain}
	D \supset D\bigl(\overline{H^\circ}\bigr) = \{0\} \oplus D^{=}\,.
\end{equation}
Now define the ``adjoint domain'' of the not densely defined operator $\overline{H^\circ}$ as 
\begin{equation}\label{adjointdomain}
	D^*\bigl(\overline{H^\circ}\bigr) = \Bigl\{ \phi \in \Hilbert ~:~ \exists \eta \in \Hilbert~\forall \psi \in D\bigl(\overline{H^\circ}\bigr): \braket{\phi \vert \overline{H^\circ} \psi}_{\Hilbert} = \braket{\eta \vert  \psi}_{\Hilbert}  \Bigr\}
\end{equation}
and note that 
\begin{equation}\label{domain3}
D \subset D^*\bigl(\overline{H^\circ}\bigr). 
\end{equation}
By \eqref{domain}, $\psi$ in \eqref{adjointdomain} is of the form $(0,f)$ with $f\in D^{=}$, so (writing $\phi^{=}=g$ and $\eta^{=}=h$)
\begin{subequations}
\begin{align}
\label{domain2}
	D^*(\overline{H^\circ})
	&=\Hilbert^{<} \oplus \Bigl\{ g \in \Hilbert^{=}  ~:~ \exists h \in \Hilbert^{=}~\forall f \in D^{=}:  \braket{g \vert H^{=} f}_{\Hilbert^{=}} = \braket{h \vert  f}_{\Hilbert^{=}}  \Bigr\}\\
 &=\Hilbert^{<} \oplus D^{=} 
\end{align}
\end{subequations}
by self-adjointness of $H^{=}$. We thus obtain the chain of inclusions
\begin{equation}\label{chain}
	\{0\} \oplus D^{=} \subset D \subset \Hilbert^{<} \oplus D^{=}. 
\end{equation}
This entails further that
\be\label{Doplus}
D=D^{<} \oplus D^{=} \text{ with }
D^{<} := \{\psi^{<}: \begin{pmatrix}\psi^{<}\\ \psi^{=} \end{pmatrix} \in D \}\,.
\ee
Indeed, writing vectors now as rows,
for any $(\psi^{<}, \psi^{=})\in D$, we have that $\psi^{<}\in D^{<}$ by definition of $D^{<}$ and $\psi^{=}\in D^{=}$ by \eqref{chain}, so $(\psi^{<}, \psi^{=})\in D^{<} \oplus D^{=}$. Conversely, if $\psi^{<}\in D^{<}$ and $\psi^{=}\in D^{=}$, then by definition of $D^{<}$ there is $\phi^{=}\in \Hilbert^{=}$ such that $(\psi^{<},\phi^{=})\in D$, but then $\phi^{=}\in D^{=}$ by \eqref{chain}, so $\psi^{=}-\phi^{=} \in D^{=}$ since $D^{=}$ is a subspace, so $(0,\psi^{=}-\phi^{=}) \in D$ by \eqref{chain}, so $(\psi^{<},\psi^{=}) = (\psi^{<},\phi^{=}) + (0,\psi^{=}-\phi^{=}) \in D$ since $D$ is a subspace.

Now we turn to the action of the operator $H$ and claim that for every $\phi = (\phi^{<}, \phi^{=}) \in D^{<} \oplus D^{=}$,
\begin{equation}\label{Hphi=}
(H\phi)^{=} = H^{=} \phi^{=}
\end{equation}
regardless of the choice of $\phi^{<}\in D^{<}$.
Indeed, from \eqref{domain3} we obtain that
\begin{align}
	&\forall \begin{pmatrix}
	\phi^{<} \\ \phi^{=}
	\end{pmatrix} = \phi \in D  ~~  \exists \begin{pmatrix}
	\eta^{<} \\ \eta^{=}
	\end{pmatrix} = \eta \in \Hilbert ~~ \forall \begin{pmatrix}
	\psi^{<} \\ \psi^{=}
	\end{pmatrix} = \psi \in D(\overline{H^\circ}):\nonumber \\
	&\braket{ (H\phi)^{=} \vert \psi^{=}}_{\Hilbert^{=}}	\stackrel{\psi^{<} = 0}{=} \braket{ H\phi \vert \psi}_{\Hilbert} \stackrel{H \text{ s.a.}}{=} \braket{\phi \vert H\psi}_{\Hilbert} \nonumber \\[3mm] 
	&\stackrel{\overline{H^\circ} \subset H}{=}
	\braket{\phi \vert \overline{H^\circ} \psi}_{\Hilbert} \stackrel{\eqref{adjointdomain}}{=}  \braket{\eta \vert \psi}_{\Hilbert} \stackrel{\psi^{<} = 0}{=} \braket{\eta^{=} \vert \psi^{=}}_{\Hilbert^{=}} \stackrel{\eqref{domain2}}{=} \braket{\phi^{=}\vert H^{=}\psi^{=}}_{\Hilbert^{=}} \nonumber \\[3mm]
	& \stackrel{H^{=} \text{ s.a.}}{=}
	\braket{H^{=}\phi^{=} \vert \psi^{=}}_{\Hilbert^{=}}.
\end{align}
By \eqref{domain}, $\braket{ (H\phi)^{=} \vert f}_{\Hilbert^{=}} = \braket{H^{=}\phi^{=} \vert f}_{\Hilbert^{=}}$ for every $f\in D^{=}$. Since $D^{=}$ is dense, \eqref{Hphi=} follows.

Now it follows further from \eqref{Hphi=} that there is an operator $H^{<}:D^{<}\to \Hilbert^{<}$ such that
\begin{equation}\label{Hphi<}
(H\phi)^{<} = H^{<} \phi^{<}
\end{equation}
regardless of $\phi^{=}$. Indeed, by setting $\phi^{=}=0$ we obtain from \eqref{Hphi=} that $(H(\phi^{<},0))^{=}=0$ and define
\be
H\begin{pmatrix}\phi^{<}\\0\end{pmatrix}=\begin{pmatrix}H^{<}\phi^{<}\\0\end{pmatrix}\,.
\ee
From \eqref{graphH} it follows that
\be
H\begin{pmatrix}0\\\phi^{=}\end{pmatrix}=\begin{pmatrix}0\\H^{=}\phi^{=}\end{pmatrix}\,.
\ee
Thus,
\be
H \begin{pmatrix}\phi^{<} \\ \phi^{=} \end{pmatrix} 
= H\begin{pmatrix}\phi^{<}\\0\end{pmatrix} + H\begin{pmatrix}0\\\phi^{=}\end{pmatrix}
= \begin{pmatrix}H^{<}\phi^{<}\\0\end{pmatrix}
+\begin{pmatrix}0\\H^{=}\phi^{=}\end{pmatrix}
= \begin{pmatrix}H^{<}\phi^{<}\\H^{=}\phi^{=}\end{pmatrix}
\ee
for all $\phi \in D =  D^{<} \oplus D^{=}$, that is, $H$ is block diagonal. 
This completes the proof.
\end{proof}

\begin{proof}[Proof of Theorem~\ref{thm:weaknogo}]
Use Theorems~\ref{thm:general} and \ref{thm:weakCoulomb} with $N_{\max}=1$, $\tilde H=H_1$ as in \eqref{H1Coulomb}, $\tilde D= C_c^\infty(\RRR^3\setminus\{\vzero\},\CCC^4)$. Since $\Hilbert^{<}=\CCC$, the only possibilities for $D^{<}$ are $D^{<}=\{0\}$ and $D^{<}=\CCC$, and since $D$ is dense, only $D^{<}=\CCC$ remains. In particular, $H^{<}$ is multiplication by some real constant $E_{00}$.
\end{proof}
 
\begin{proof}[Proof of Theorem~\ref{thm:sa}]
We begin by reviewing the well known decomposition of Hilbert space $\Hilbert^{(1)}=L^2(\RRR^3,\CCC^4)$ in terms of the $\Phi^{\pm}_{m_j,\kappa_j}$. 
By passing to spherical coordinates and denoting with $d^2\vomega$ the surface measure of the unit sphere $\SSS^2$, we obtain the canonical isomorphism
\begin{equation}
	U : \Hilbert^{(1)} \to L^2((0,\infty),\CCC, dr) \otimes L^2(\SSS^2, \CCC^4, d^2\vomega)
\end{equation}
by setting for each $\psi \in \Hilbert^{(1)}$
\begin{equation}\label{Udef}
(U\psi)(r,\vomega) = r \psi(r\vomega). 
\end{equation}
Under this transformation, the Dirac-Coulomb operator $H_1$ as in \eqref{H1Coulomb} takes the form \cite[p.~125]{Tha91}
\begin{equation}
UH_1U^{\dagger} = -i\alpha_r \left(\partial_r +\frac{1}{r}- \frac{1}{r}\beta K\right)+m\beta + \frac{q}{r}, 
\end{equation}
where $\alpha_r = \ve_r \cdot \valpha$ is the radial component of $\valpha$, $K=\beta(2\vS\cdot \vL +1)$ is the spin-orbit operator consisting of the spin operator $\vS = -\frac{i}{4}\valpha \times \valpha$ and the angular momentum operator $\vL = \vx \times (-i \nabla)$. Denoting with $\vJ=\vL+\vS$ the total angular momentum, one can show \cite{Tha91} that $K$ commutes with $\vJ^2$ and with the third component $J_3$ of $\vJ$. The $\Phi_{m_j, \kappa_j}^{\pm}$ form a joint eigenbasis of $\vJ^2, K, J_3$ with eigenvalues $j(j+1)$, $\kappa_j$, and $m_j$, respectively, thus providing the following orthogonal decomposition: 
\begin{equation}
L^2(\SSS^2,\CCC^4,d\Omega) = \bigoplus_{j \in \mathbb{N}_0+\frac{1}{2}} \bigoplus_{m_j = -j}^{j} \bigoplus_{\kappa_j= \pm(j+\frac{1}{2})} \Kilbert_{m_j\kappa_j} 
\end{equation}
with
\be\label{Kilbertdef}
\Kilbert_{m_j\kappa_j}=\mathrm{span}(\Phi^+_{m_j\kappa_j}, \Phi^-_{m_j\kappa_j})\,.
\ee
In this basis, we have that \cite[Lemma 4.13]{Tha91}
\be\label{alphaPhi}
\alpha_r = \begin{pmatrix}
0 & -i \\ 
i & 0
\end{pmatrix}, \quad \quad \beta = \begin{pmatrix}
1 & 0 \\ 
0 & -1
\end{pmatrix}, \quad \quad i\alpha_r \beta = \begin{pmatrix}
0 & -1 \\ 
-1 & 0
\end{pmatrix},
\ee
and we quote the following fact:

\begin{lem}\label{lem:H1Phi} \cite[Thm.~4.14]{Tha91}
$U$ as in \eqref{Udef} maps $C_c^\infty(\Q_1,\CCC^4)$ to
\be
\bigoplus_{j ,m_j,\kappa_j} C_c^\infty\bigl((0,\infty),\CCC\bigr) \otimes \Kilbert_{m_j\kappa_j}\,,
\ee
and $UH_1U^\dagger$ is block diagonal with blocks
\begin{equation}\label{hdef}
h_{m_j, \kappa_j} = 
\begin{pmatrix}
m+\frac{q}{r} & -\partial_r + \frac{\kappa_j}{r} \\[2mm] 
\partial_r + \frac{\kappa_j}{r} & -m + \frac{q}{r}
\end{pmatrix}\,.
\end{equation}
\end{lem}

In the following, we will say ``angular momentum sector'' to $\Kilbert_{m_j\kappa_j}$. Our construction of the IBC Hamiltonian $H$ proceeds for each angular momentum sector separately. For one sector, the one chosen in \eqref{mjkappaj}, we will couple $h_{\tilde m_j \tilde\kappa_j}$ to the 0-particle sector of our mini-Fock space $\Hilbert$; all other angular momentum sectors will decouple. That is, $H$ will be block diagonal relative to the sum decomposition
\be\label{Hsummands}
\Hilbert \cong \widehat\Hilbert \oplus \bigoplus_{(j,m_j,\kappa_j) \neq (\tilde\jmath,\tilde m_j,\tilde \kappa_j)} L^2((0,\infty))\otimes \Kilbert_{m_j \kappa_j}
\ee
(note that $j$ is determined by $\kappa_j$ through $j=|\kappa_j|-\tfrac12$),
but not relative to
\be\label{Hilbertcouple}
\widehat\Hilbert=\Hilbert^{(0)}\oplus L^2((0,\infty)) \otimes \Kilbert_{\tilde m_j \tilde\kappa_j}\,.
\ee
To this end, we need a self-adjoint extension for every $h_{ m_j \kappa_j}$ with $(m_j,\kappa_j) \neq (\tilde m_j,\tilde \kappa_j)$.

\begin{lem}\label{lem:esa} \cite{Hog12}, \cite[Prop.s 1.2 and 2.2--2.4]{GM19}
The operator $(h_{m_j \kappa_j},C_c^\infty((0,\infty))\otimes \Kilbert_{m_j \kappa_j})$ is essentially self-adjoint if and only if $q^2\leq \kappa_j^2-\tfrac14$. As a consequence, for $\sqrt{3}/2 < |q| < 1$, the only angular momentum sectors for which $h_{m_j \kappa_j}$ is not essentially self-adjoint are those mentioned in \eqref{mjkappaj}. Furthermore, for $\sqrt{3}/2 < |q| < 1$, those sectors that are not essentially self-adjoint possess, among an infinitude of self-adjoint extensions, a distinguished one $h_{D m_j  \kappa_j}$ that is uniquely characterized by the property that for all functions $\phi$ in the domain, kinetic and potential energy are separately finite, $\| \, |h_{D m_j \kappa_j}-\tfrac{q}{r}|^{1/2}\phi\|<\infty$ and $\|(\tfrac{|q|}{r})^{1/2}\phi\|<\infty$. Functions $\phi$ in the domain of the distinguished extension obey the asymptotics
\be\label{Dasymp}
\phi(r,\vomega) = c_{+ m_j \kappa_j} \, f^+_{B m_j \kappa_j} (\vomega)\, r^B + o(r^{1/2})
\ee
as $r\searrow 0$.
\end{lem}

For the three sectors in $\sA$ mentioned in \eqref{mjkappaj} but different from $(\tilde m_j,\tilde\kappa_j)$, we choose the distinguished extension, and for those not mentioned there (i.e., with $|\kappa_j|\geq 2$), the extension is unique. These extensions can be combined to form an extension of $H^\circ$ on all summands but the first in \eqref{Hsummands}. So it remains to construct the block $\widehat H$ of $H$ acting in $\widehat\Hilbert$, which is where the coupling between particle sectors takes place. Correspondingly, 
\be
D:= \widehat D \oplus D_\mathrm{distinguished} \subset \widehat \Hilbert \oplus \widehat\Hilbert\:{}^\perp = \Hilbert
\ee
with $\widehat D$ to be determined. For brevity, we set
\be
\Kilbert:= \Kilbert_{\tilde m_j \tilde\kappa_j} \text{ and } h:= h_{\tilde m_j \tilde\kappa_j}\,.
\ee
We regard $h$ as defined on $C_c^\infty((0,\infty))\otimes \Kilbert$. Another known fact:

\begin{lem}\label{lem:h*} \cite[Thm.~2.6]{GM19}
Every function $\phi$ in the domain $D(h^*)$ of the adjoint $h^*$  of $h$ obeys, as $r\searrow 0$, the asymptotics\footnote{Note that the meaning of symbols such as $\psi^{(1)}$ differs by a factor of $r$ (coming from \eqref{Udef}) depending on whether we consider the left or the right-hand side of the isomorphism \eqref{Hsummands}. This is why the exponents are $\pm B$ in \eqref{asymp} but $-1\pm B$ in \eqref{short}.}
\be\label{asymp}
\phi(r,\vomega) = c_-\, f^{-}_{\tilde m_j \tilde\kappa_j}(\vomega)\, r^{-B} +  c_{+}\, f^+_{\tilde m_j \tilde\kappa_j}(\vomega) \, r^B + o(r^{1/2})
\ee
with $f^{\pm}$ as in \eqref{fdef}; the coefficients $c_-,c_+=c_{+ \tilde m_j \tilde\kappa_j} \in \CCC$ are uniquely determined by $\phi$, and all combinations $(c_-,c_{+})\in\CCC^2$ occur for some $\phi$.
\end{lem}

For higher angular momentum sectors ($|\kappa_j|\geq 2$), functions in the domain of the unique self-adjoint extension of $h_{m_j\kappa_j}$ are $o(r^{1/2})$ (because the extension is $\overline{h_{m_j\kappa_j}}$, and the domain of that contains only $o(r^{1/2})$ functions by \cite[Prop.~2.4]{GM19}) and thus contribute only $o(|\vx|^{-1/2})$ in \eqref{short}.

We will directly consider the more general form \eqref{IBCgen} of the IBC, which contains \eqref{IBC} as a special case. We define the domain $\widehat D$ of $\widehat H$ as follows, containing functions satisfying the IBC \eqref{IBCgen}:\be
\widehat D:= \Bigl\{ (\psi^{(0)},\psi^{(1)})\in \widehat\Hilbert: \psi^{(1)}\in D(h^*) \text{ and } \eqref{IBCgen} \Bigr\}\,.
\ee
We define the action of $\widehat H$ (according to \eqref{Hact1} and \eqref{Hact0gen}) as 
\begin{subequations}\label{hatHdef}
\begin{align}
(\widehat H\psi)^{(0)}&= g^* \,(a_3\, c_- + a_4 \, c_{+ \tilde m_j \tilde\kappa_j}) \\
(\widehat H\psi)^{(1)}&=h^* \psi^{(1)}\,.
\end{align}
\end{subequations}

$\widehat D$ is dense in $\widehat\Hilbert$ because (i)~in $D(h^*)$ there exist functions with arbitrary values of $c_-$ and $c_{+}$, so for $\psi\in\widehat D$ any desired complex number can occur as $\psi^{(0)}$; (ii)~since $D(h)=C_c^\infty((0,\infty))\otimes \Kilbert$ is contained in $D(h^*)$ and dense in $L^2((0,\infty))\otimes \Kilbert$, the set $\{(\psi^{(0)},\psi^{(1)}+\phi):\phi\in D(h)\} \subset \widehat D$ is dense in $\{\psi^{(0)}\}\oplus L^2((0,\infty))\otimes \Kilbert$. Together, (i) and (ii) imply that $\widehat D$ is dense.

We now prove that $\widehat H$ is a symmetric operator on $\widehat D$. 
Let $\psi \in \widehat D$ with short distance coefficients $c_-$ and $c_+$, and $ \phi \in \widehat D$ with short distance coefficients $d_-$ and $d_+$. Then
\begin{align}
&\braket{\phi, H\psi}_{\widehat\Hilbert} - \braket{H\phi, \psi}_{\widehat\Hilbert} \nonumber\\
= & \braket{\phi^{(0)}, g^* \,(a_3\, c_- + a_4 \, c_{+})}_{\mathbb{C}}+\braket{\phi^{(1)}, h^* \psi^{(1)}}_{L^2((0,\infty),\Kilbert)} \nonumber\\ 
&- \braket{g^* \,(a_3\, d_- + a_4 \, d_+),\psi^{(0)} }_{\mathbb{C}}
- \braket{h^* \phi^{(1)}, \psi^{(1)}}_{L^2((0,\infty),\Kilbert)} \\ 
= &\: \braket{a_1 d_-+a_2 d_+, a_3\, c_- + a_4 \, c_{+}}_{\mathbb{C}} - \braket{a_3 d_-+ a_4 d_+, a_1 c_- + a_2 c_+}_{\mathbb{C}} \nonumber\\ 
&+ \braket{\phi^{(1)}, h^* \psi^{(1)}}_{L^2((0,\infty),\Kilbert)}
- \braket{h^* \phi^{(1)}, \psi^{(1)}}_{L^2((0,\infty),\Kilbert)}~.\label{Hsym1}
\end{align}
Now note that $h^*$ can be written as $h^* = -i\alpha_r \partial_r + M(r)$, where $M(r)$ is an Hermitian operator $\Kilbert\to\Kilbert$ for each $r>0$. Thus, the last two terms of \eqref{Hsym1} can be written as
\begin{align}
&\braket{\phi^{(1)}, h^* \psi^{(1)}}_{L^2((0,\infty),\Kilbert)}
- \braket{h^* \phi^{(1)}, \psi^{(1)}}_{L^2((0,\infty),\Kilbert)} \nonumber\\
=& \int_{0}^{\infty} dr \Bigl[ \braket{\phi^{(1)}(r), (-i\alpha_r) \partial_r \psi^{(1)}(r)}_\Kilbert+ \braket{\partial_r\phi^{(1)}(r), (-i\alpha_r)\psi^{(1)}(r)}_\Kilbert\Bigr] \\ 
=& \int_{0}^{\infty} dr \: \partial_r \braket{\phi^{(1)}(r), (-i\alpha_r) \psi^{(1)}(r)}_\Kilbert \\ 
=& ~\lim_{r\searrow 0} \braket{\phi^{(1)}(r), (+i\alpha_r) \psi^{(1)}(r)}_\Kilbert\label{Hsym1b}\\
=& ~\lim_{r\searrow 0} \braket{d_- f^- r^{-B}+d_+ f^+ r^B + o(r^{1/2}),  c_- (i\alpha_r) f^- r^{-B}+c_+ (i\alpha_r)f^+ r^B + o(r^{1/2})}_\Kilbert\\
=& ~\lim_{r\searrow 0} \Bigl[d_-^* c_- \braket{f^-,i\alpha_rf^-}_\Kilbert r^{-2B} \nonumber\\
&+ \Bigl(d^*_- c_+  \braket{f^-,i\alpha_rf^+}_\Kilbert
+ d^*_+ c_- \braket{f^+,i\alpha_rf^-}_\Kilbert\Bigr) r^0 +  o(r^0)\Bigr]\,.\label{Hsym2}
\end{align}
Since in the orthonormal basis $\Phi^{\pm}=\Phi^{\pm}_{\tilde m_j \tilde\kappa_j}$ of $\Kilbert$, the explicit form of $\alpha_r$ is given by \eqref{alphaPhi}, any vector $v=x \Phi^+ + y \Phi^-\in \Kilbert$ with coefficients $x,y\in\CCC$ has
\be
\braket{v,i\alpha_r v}_\Kilbert = (x^*~~y^*) \begin{pmatrix} 0&1\\-1&0 \end{pmatrix} \begin{pmatrix} x\\y \end{pmatrix} = x^*y-y^*x = 2\, \Im(x^*y)\,.
\ee
Since by \eqref{fdef}, $f^{\pm}$ has real coefficients relative to $\Phi^{\pm}$, we have that
\be\label{falphaf}
\braket{f^{\pm},i\alpha_r f^{\pm}}_\Kilbert=0.
\ee
Thus, the $r^{-2B}$ term in \eqref{Hsym2} vanishes. Moreover, by \eqref{fdef} and \eqref{alphaPhi} again,
\be\label{f-alphaf+}
\braket{f^-,i\alpha_r f^+}_\Kilbert= -4B(1+q)\,,~~~~~~
\braket{f^+,i\alpha_r f^-}_\Kilbert=4B(1+q)\,.
\ee
Thus,
\begin{align}
&\braket{\phi^{(1)}, h^* \psi^{(1)}}_{L^2((0,\infty),\Kilbert)}
- \braket{h^* \phi^{(1)}, \psi^{(1)}}_{L^2((0,\infty),\Kilbert)} \nonumber\\
&= d^*_+ c_-  4B(1+q)- d^*_- c_+ 4B(1+q)\,.\label{Hsym3}
\end{align}
And therefore, putting together \eqref{Hsym1} and \eqref{Hsym3},
\begin{align}
\braket{\phi, H\psi}_{\Hilbert} - \braket{H\phi, \psi}_{\Hilbert} 
&= a_1 a_3 d^*_- c_- +a_2 a_3 d^*_+ c_- + a_1 a_4 d^*_- c_+ + a_2 a_4 d^*_+ c_+ \nonumber\\[2mm]
&~~~- a_1 a_3 d^*_- c_- - a_2 a_3 d^*_- c_+ - a_1 a_4 d^*_+ c_- - a_2 a_4 d^*_+ c_+ \nonumber\\
&~~~+ \Bigl( d^*_+ c_- - d^*_- c_+\Bigr) 4B(1+q)\\
&= (d^*_+c_- - d^*_- c_+) (-a_1 a_4 + a_2 a_3 + 4B(1+q))\\[2mm]
&=0\label{lastsymmetry}
\end{align}
if \eqref{acond} holds.
	
In order to see that $\widehat H$ is also self-adjoint, it remains to verify that $\widehat D= D\bigl( {\widehat{H}}^*\bigr)$. To this end, we first note that
\begin{equation}
\widehat D\subseteq D\bigl({\widehat{H}}^*\bigr) \subseteq \mathbb{C} \oplus D(h^*)
\end{equation}
Any given $\phi \in \mathbb{C} \oplus D(h^*)$ lies in $D\bigl({\widehat{H}}^*\bigr)$ if and only if there is $\eta \in \widehat\Hilbert$ such that for every $\psi \in \widehat D$,
\begin{equation}\label{domcond}
 \braket{\eta, \psi}_{\widehat\Hilbert} = \braket{\phi, \widehat H\psi}_{\widehat\Hilbert}.
\end{equation}
If we write again $c_-,c_+$ for the short-distance coefficients of $\psi^{(1)}$ and $d_-,d_+$ for those of $\phi^{(1)}$, then the condition \eqref{domcond} is equivalent to
\begin{align}
\braket{\eta,\psi}_{\widehat{\Hilbert}}
&~~~=~~~ \braket{\phi^{(0)},(\widehat H\psi)^{(0)}}_{\CCC} + \braket{\phi^{(1)},h^*\psi^{(1)}}_{L^2((0,\infty),\Kilbert)}\\
&\stackrel{\eqref{Hact0gen},\eqref{Hsym3}}{=} \phi^{(0)*} \, g^*(a_3 c_-+a_4 c_+)\nonumber\\
&~~~~~~~~~+\braket{h^*\phi^{(1)},\psi^{(1)}}_{L^2((0,\infty),\Kilbert)}-(d^*_-c_+-d^*_+ c_-)4B(1+q)\\
&~~\stackrel{\eqref{acond}}{=}~~ \bigl[g\phi^{(0)} \bigr]^*(a_3 c_-+a_4 c_+)\nonumber\\
&~~~~~~~~~+\braket{h^*\phi^{(1)},\psi^{(1)}}_{L^2((0,\infty),\Kilbert)}-
\begin{vmatrix}d^*_-&d^*_+\\ c_-&c_{+}\end{vmatrix} \:
\begin{vmatrix}a_1&a_3\\a_2&a_4\end{vmatrix}\\
&~~~=~~~ \bigl[g\phi^{(0)} \bigr]^*(a_3 c_-+a_4 c_+)\nonumber\\
&~~~~~~~~~+\braket{h^*\phi^{(1)},\psi^{(1)}}_{L^2((0,\infty),\Kilbert)}-
\begin{vmatrix}d^*_-a_1+d^*_+ a_2 & d^*_-a_3+d^*_+ a_4\\
c_-a_1+c_+ a_2 & c_-a_3+c_+ a_4\end{vmatrix}\\
&~~~=~~~ \Bigl[-(a_1d_-+a_2d_+)+g\phi^{(0)} \Bigr]^* (a_3 c_-+a_4 c_+)\nonumber\\
&~~~~~~~~~+\braket{h^*\phi^{(1)},\psi^{(1)}}_{L^2((0,\infty),\Kilbert)}+
\bigl(a_3 d^*_-+a_4 d^*_+ \bigr) \bigr(a_1 c_-+a_2 c_+  \bigr)\\
&~~~=~~~ \Bigl[-(a_1d_-+a_2d_+)+g\phi^{(0)} \Bigr]^* (a_3 c_-+a_4 c_+)\nonumber\\
&~~~~~~~~~+\braket{h^*\phi^{(1)},\psi^{(1)}}_{L^2((0,\infty),\Kilbert)}+\braket{g^*(a_3 d_-+a_4 d_+) ,\psi^{(0)}}_{\CCC}~.\label{lastsa}
\end{align}
The only way this can be true for all $\psi\in \widehat D$ is that \begin{subequations}
\begin{align}
\eta^{(0)}&=g^*(a_3d_-+a_4d_+),\\
\eta^{(1)}&=h^* \phi^{(1)},
\end{align}
\end{subequations}
and the product in the first line of \eqref{lastsa} vanishes. For one thing, we obtain from this an expression for ${\widehat{H}}^*\phi=\eta$. Moreover, since the product in the first line of \eqref{lastsa} consists of one factor depending on $\phi$ and one depending on $\psi$, and since the product needs to vanish for every $\psi\in \widehat D$ but the second factor will not, the first factor has to vanish for every $\phi\in D({\widehat H}^*)$.
As a consequence, $\phi$ also needs to satisfy the IBC \eqref{IBCgen}, i.e.,
\begin{equation}
	a_1 \, d_-+a_2 \, d_+ = g \phi^{(0)},
\end{equation}
so that we arrive at
\begin{equation}
	\widehat D = D\bigl({\widehat{H}}^*\bigr)
\end{equation}
and $\widehat H$ (and thus $H$) is self-adjoint. This completes the proof of Theorem~\ref{thm:sa}.
\end{proof}

\begin{rmk}
Here is an alternative argument for the part after \eqref{lastsymmetry}, after it has been shown that $H$ is symmetric. While this alternative argument does not show that $H$ is self-adjoint, it shows that $H$ possesses a self-adjoint extension, which suffices for proving Theorem~\ref{thm:sa}. The argument is based on the von Neumann theorem about conjugations \cite[Thm.~X.3]{reedsimon2}, which asserts the following: {\it Let $\Hilbert$ be a Hilbert space. An antilinear map $C: \Hilbert \to \Hilbert$ (i.e., such that $C(\alpha \phi + \beta \psi) = \alpha^* C \phi + \beta^* C \psi$) is called a conjugation if it is norm-preserving and $C^2 = I$. Let $(A,D(A))$ be a densely defined symmetric operator and suppose that there exists a conjugation $C$ with $C(D(A))\subseteq D(A)$ and $AC=CA$. Then $A$ has equal deficiency indices and therefore has self-adjoint extensions.}

In our case, we use that the differential expression \eqref{hdef} for $h$ has only real entries. We assume for the argument that the coupling constant $g$ is real; if this is not the case, it can be arranged through a unitary transformation of $\Hilbert$ that merely changes the phase of $\psi^{(0)}$ by the phase of $g$. So, we take $C:\widehat\Hilbert\to \widehat\Hilbert$ to be complex conjugation of $\psi^{(0)}$ and of the coefficients of $\psi^{(1)}\in L^2((0,\infty))\otimes \Kilbert$ relative to $\Phi^{\pm}_{\tilde m_j \tilde\kappa_j}$. Since the coefficients $a_1,\ldots,a_4$ are real, $H$ commutes with $C$, and the von Neumann theorem applies.
\end{rmk}

\begin{proof}[Proof of Lemma~\ref{lem:rotation}]
Suppose $\psi$ was a nonzero vector in $L^2(\RRR^3,\CCC^4)$ invariant under the representation of the covering group of $SO(3)$. 
A rotation (about any axis) through 360 degrees is a particular element $g$ of the covering group that acts on vectors $\vx$ in $\RRR^3$ as the identity and on spinors as $-I$ with $I$ the identity. Thus, $g$ maps $\psi$ to $-\psi$, and if $\psi$ is invariant, it must vanish.
\end{proof}

\begin{proof}[Proof of Theorem~\ref{thm:rotation}]
If $H$ were a self-adjoint extension of $H^\circ$ that involves particle creation, then the initial Fock vector $\psi_0=(\psi_0^{(0)},\psi_0^{(1)})=(1,0)$ (i.e., the Fock vacuum) would evolve by some time $t>0$ to a non-vacuum state, i.e., one with nonzero $\psi_t^{(1)}$, that would be invariant under rotations. Since rotations do not mix particle number sectors, also $\psi_t^{(1)}$ by itself would be invariant under rotations. However, by Lemma~\ref{lem:rotation} no such state exists. 
\end{proof}

\bigskip

\noindent\textit{Acknowledgments.} We thank Stefan Teufel for helpful discussion and advice. J.H.~gratefully acknowledges partial financial support by the ERC Advanced Grant ``RMTBeyond" No.~101020331.

\end{document}